\documentclass{llncs}

\usepackage{fullpage}
\usepackage{amssymb}
\usepackage{xspace}
\usepackage[english]{babel}
\usepackage{amsmath}
\usepackage{amsopn}
\usepackage{latexsym}
\usepackage{textcomp}
\usepackage{epsfig}
\usepackage{subfigure}
\usepackage{color}
\usepackage{times}
\usepackage{enumerate}
\usepackage{tabularx}
\usepackage[plainpages=false]{hyperref}
\usepackage[figure,table]{hypcap}
\usepackage{multirow}
\usepackage{todonotes}

\newcommand{\remove}[1]{}

\DeclareMathOperator{\CG}{CR}

\title{Properties and Complexity of Fan-Planarity 
\thanks{Research supported in part by the MIUR project AMANDA ``Algorithmics for MAssive and Networked DAta'', prot. 2012C4E3KT\_001.}
}

\author{Carla Binucci\inst{1},
Emilio Di Giacomo\inst{1},
Walter Didimo\inst{1},
Fabrizio Montecchiani\inst{1},\\
Maurizio Patrignani\inst{2},
Ioannis G. Tollis\inst{3}
}

\date{}

\institute{
Universit\`a degli Studi di Perugia, Italy\\
\email{\{carla.binucci,emilio.digiacomo,\\walter.didimo,fabrizio.montecchiani\}@unipg.it} \and
Universit\`a  Roma Tre, Italy\\
\email{patrigna@dia.uniroma3.it} \and
Univ. of Crete and Institute of Computer Science-FORTH, Greece\\
\email{tollis@ics.forth.gr}
}

\begin{document}
\maketitle

\begin{abstract}
In a \emph{fan-planar drawing} of a graph an edge can cross only edges with a common end-vertex. Fan-planar drawings have been recently introduced by Kaufmann and Ueckerdt, who proved that every $n$-vertex fan-planar drawing has at most $5n-10$ edges, and that this bound is tight for $n \geq 20$. We extend their result, both from the combinatorial and the algorithmic point of view. We prove tight bounds on the density of constrained versions of fan-planar drawings and study the relationship between fan-planarity and $k$-planarity. Furthermore, we prove that deciding whether a graph admits a fan-planar drawing in the variable embedding setting is NP-complete.
\end{abstract}

\section{Introduction}\label{se:introduction}

There is a growing interest in the study of non-planar drawings of graphs with forbidden crossing configurations. The idea is to relax the planarity constraint by allowing edge crossings that do not affect too much the drawing readability. Among the most popular types of non-planar drawings studied so far we mention: \emph{$k$-planar drawings}, in which an edge can have at most $k$ crossings (see, e.g.,~\cite{DBLP:conf/gd/AlamBK13,DBLP:conf/gd/AuerBBGHNR13,DBLP:conf/gd/AuerBGH12,DBLP:conf/gd/BrandenburgEGGHR12,DBLP:journals/comgeo/DiGiacomoDLM13,DBLP:journals/ipl/Didimo13,DBLP:journals/dam/EadesL13,DBLP:conf/gd/HongEKLSS13,DBLP:conf/cocoon/HongELP12,DBLP:journals/jgt/KorzhikM13,DBLP:journals/combinatorica/PachT97,DBLP:journals/siamdm/Suzuki10}); \emph{$k$-quasi-planar drawings}, which do not contain $k$ mutually crossing edges (see, e.g.,~\cite{DBLP:journals/dcg/Ackerman09,DBLP:journals/jct/AckermanT07,DBLP:journals/combinatorica/AgarwalAPPS97,DBLP:conf/wg/DiGiacomoDLM12,DBLP:journals/siamdm/FoxPS13,DBLP:journals/dcg/Valtr98}); \emph{RAC drawings}, where edges can cross only at right angles (see, e.g.,~\cite{DBLP:journals/tcs/DidimoEL11} and~\cite{dl-cargd-12} for a survey); \emph{ACE$_\alpha$ drawings}~\cite{DBLP:journals/siamdm/AckermanFT12} and \emph{ACL$_\alpha$ drawings}~\cite{DBLP:conf/s-egc/AngeliniBDFHKLL11,DBLP:journals/mst/DiGiacomoDLM11,DBLP:journals/cjtcs/DujmovicGMW11}, which are generalizations of RAC drawings; namely, in an ACE$_\alpha$ drawing edges can cross only at an angle that is \emph{exactly} $\alpha$ ($\alpha \in (0,\pi/2]$), while in an ACL$_\alpha$ drawing edges can cross only at angles that are \emph{at least} $\alpha$ (see also~\cite{dl-cargd-12}); \emph{fan-crossing free drawings}, where there cannot be an edge that crosses two other edges with a common end-vertex~\cite{DBLP:conf/isaac/CheongHKK13}.  

Given a desired type $T$ of non-planar drawing with forbidden crossing configurations, a classical combinatorial problem is to establish bounds on the maximum number of edges that a drawing of type $T$ can have; this problem is usually dubbed as a Tur{\'a}n-type problem, and several tight bounds have been proved for the types of drawings mentioned above, both for straight-line and for polyline edges (see, e.g.,~\cite{DBLP:journals/dcg/Ackerman09,DBLP:journals/siamdm/AckermanFT12,DBLP:journals/combinatorica/AgarwalAPPS97,DBLP:conf/gd/BrandenburgEGGHR12,DBLP:conf/isaac/CheongHKK13,DBLP:journals/ipl/Didimo13,DBLP:journals/tcs/DidimoEL11,DBLP:journals/cjtcs/DujmovicGMW11,DBLP:journals/siamdm/FoxPS13,DBLP:journals/combinatorica/PachT97,DBLP:journals/dcg/Valtr98}). From the algorithmic point of view, the complexity of testing whether a graph $G$ admits a drawing of type $T$ is one of the most interesting. Also for this problem several results have been shown, both in the variable and in the fixed embedding setting (see, e.g.,~\cite{DBLP:conf/gd/AuerBBGHNR13,DBLP:journals/ijcga/DehkordiE12,DBLP:journals/algorithmica/DiGiacomoDEL14,DBLP:journals/algorithmica/GrigorievB07,DBLP:conf/gd/HongEKLSS13,DBLP:journals/jgt/KorzhikM13}).

\smallskip
In this paper we investigate \emph{fan-planar drawings} of graphs, in which an edge cannot cross two independent edges, i.e., an edge can cross several edges provided that they have a common end-vertex. Fan-planar drawings have been recently introduced by Kaufmann and Ueckerdt~\cite{DBLP:journals/corr/KaufmannU14}; they proved that every $n$-vertex graph without loops and multiple edges that admits a fan-planar drawing has at most $5n-10$ edges, and that this bound is tight for $n \geq 20$. Fan-planar drawings are on the opposite side of fan-crossing free drawings mentioned above.
Besides their intrinsic theoretical interest, we find that fan-planar drawings can be used as a basis for generating drawings with few edge crossings in a confluent drawing style (see, e.g.,~\cite{DBLP:journals/jgaa/DickersonEGM05,DBLP:journals/algorithmica/EppsteinGM07}). For example, Figure~\ref{fi:intro-a} shows a fan-planar drawing $\Gamma$ with 12 crossings; Figure~\ref{fi:intro-b} shows a new drawing with just 3 crossings obtained from $\Gamma$ by bundling crossing ``fans''.

\begin{figure}
\centering
\subfigure[]{\label{fi:intro-a}\includegraphics[scale=0.5]{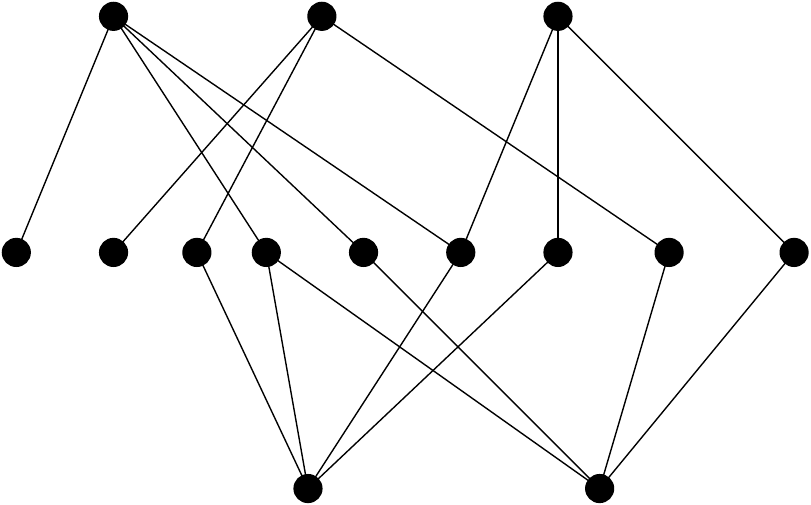}}
\hspace{1cm}
\subfigure[]{\label{fi:intro-b}\includegraphics[scale=0.5]{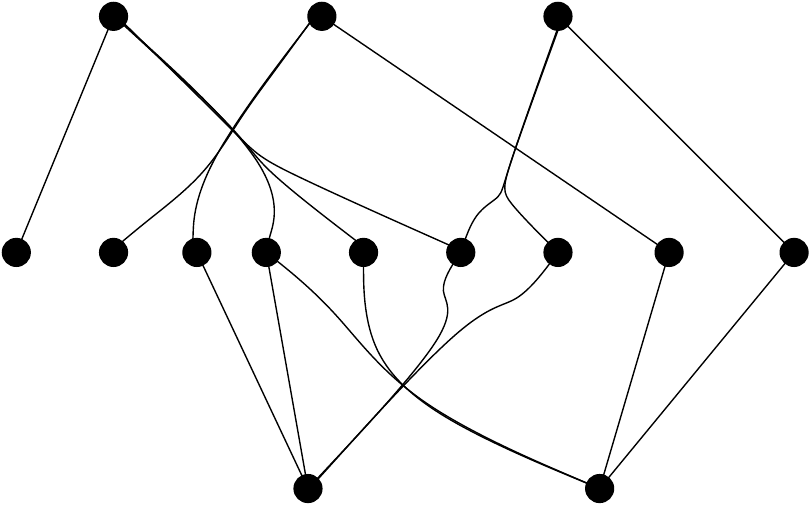}}
\caption{(a) A fan-planar drawing with 12 crossings; (b) A confluent drawing with 3 crossings.}
\end{figure}

We prove both combinatorial properties and complexity results related to fan-planar drawings of graphs. The main contributions of our work are as follows:

\smallskip\noindent{$(i)$} We study the density of constrained versions of fan-planar drawings, namely \emph{outer fan-planar drawings}, where all vertices must lie on the external boundary of the drawing (Section~\ref{se:outerfan}), and \emph{$2$-layer fan-planar drawings}, where vertices are placed on two distinct horizontal lines and edges must be vertically monotone lines. We prove tight bounds for the edge density of these drawings. More precisely, we show that $n$-vertex outer fan-planar drawings have at most $3n-5$ edges and that this bound is tight for $n \geq 5$. Also, we prove that $n$-vertex $2$-layer fan-planar drawings have at most $2n-4$ edges and that this bound is also tight for $n \geq 3$. We remark that outer and $2$-layer non-planar drawings have been previously studied   in the 1-planarity setting~\cite{DBLP:conf/gd/AuerBBGHNR13,DBLP:journals/ipl/Didimo13,DBLP:conf/gd/HongEKLSS13} and in the RAC planarity setting~\cite{DBLP:journals/ijcga/DehkordiE12,DBLP:journals/algorithmica/DiGiacomoDEL14}. 

\smallskip\noindent{$(ii)$} Since general fan-planar drawings have at most $5n-10$ edges and the same bound holds for $2$-planar drawable graphs~\cite{DBLP:journals/combinatorica/PachT97}, we investigate the relationship between these two graph classes. More in general, we are able to prove that in fact for any $k \geq 2$ there exist graphs that are fan-planar drawable but not $k$-planar drawable, and vice-versa (Section~\ref{se:2planar}).

\smallskip\noindent{$(iii)$} Finally, exploiting some ingredients used for proving the previously mentioned results, we show that testing whether a graph admits a fan-planar drawing in the variable embedding setting is NP-complete (Section~\ref{se:hardness}).

\medskip 
Preliminary definitions are given in Section~\ref{se:preliminaries}. Conclusions and open problems can be found in Section~\ref{se:conclusions}.    
 
\section{Preliminary definitions and results}\label{se:preliminaries}
A \emph{drawing} $\Gamma$ of a graph $G$ maps each vertex to a distinct point of the plane and each edge to a simple Jordan arc between the points corresponding to the end-vertices of the edge. For a subgraph $G'$ of $G$, we denote by $\Gamma[G']$ the restriction of $\Gamma$ to $G'$. Throughout the paper we consider only \emph{simple graphs}, i.e., graphs with neither multiple edges nor self-loops; also we only consider \emph{simple drawings}, i.e., drawings such that the arcs representing two edges have at most one point in common, which is either a common end-vertex or a common interior point where the two arcs properly cross each other.

For each vertex $v$ of $G$, the set of edges incident to $v$ is called the \emph{fan of $v$}. Clearly, each edge $(u,v)$ of $G$ belongs to the fan of $u$ and to the fan of $v$ at the same time. Two edges that do not share a vertex are called \emph{independent edges}; two independent edges always belong to distinct fans.  
A \emph{fan-planar drawing} $\Gamma$ of $G$, is a drawing of $G$ such that: $(a)$ no edge is crossed by two independent edges; $(b)$ there are not two adjacent edges $(u,v)$, $(u,w)$ that cross an edge $e$ from different 
``sides'' while moving from $u$ to $v$ and from $u$ to $w$. The forbidden configurations $(a)$ and $(b)$ are depicted in Figure~\ref{fi:forbidden-conf-1} and Figure~\ref{fi:forbidden-conf-2}, respectively.
Figures~\ref{fi:allowed-conf-1} and~\ref{fi:allowed-conf-2} show two allowed configurations of a fan-planar drawing. A \emph{fan-planar graph} is a graph that admits a fan-planar drawing. 

\begin{figure}
\centering
\subfigure[]{\label{fi:forbidden-conf-1}\includegraphics[scale=0.5]{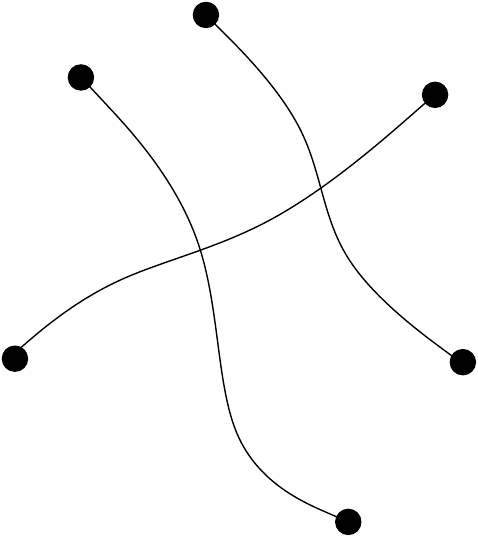}}
\hspace{1cm}
\subfigure[]{\label{fi:forbidden-conf-2}\includegraphics[scale=0.5]{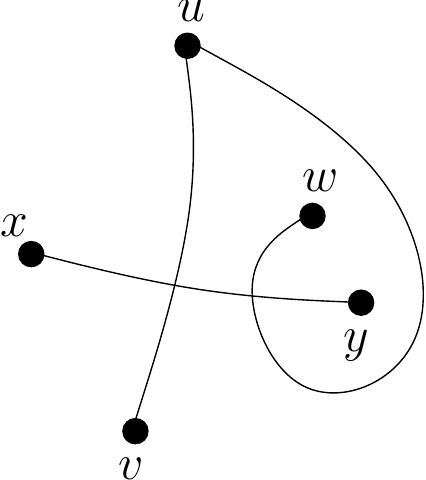}}
\hspace{1cm}
\subfigure[]{\label{fi:allowed-conf-1}\includegraphics[scale=0.5]{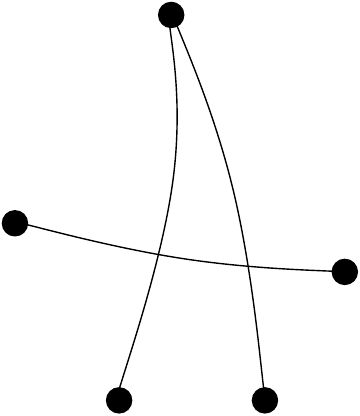}}
\hspace{1cm}
\subfigure[]{\label{fi:allowed-conf-2}\includegraphics[scale=0.5]{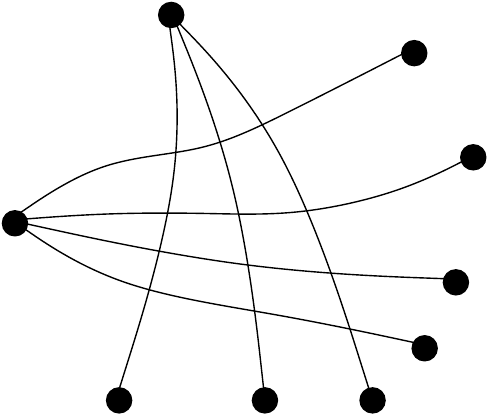}}
\caption{(a)-(b) Forbidden configurations in a fan-planar drawing: (b) Allowed configurations in a fan-planar drawing.}
\end{figure}

The following property is an immediate consequence of the definition of fan-planar drawings.

\begin{property}\label{pr:3-mutually}
A fan-planar drawing does not contain 3-mutually crossing edges.
\end{property}

Let $\Gamma$ be a non-planar drawing of $G$; the \emph{planar enhancement} $\Gamma'$ of $\Gamma$ is the drawing obtained from $\Gamma$ by replacing each crossing point with a dummy vertex. The boundary of each face $f'$ of $\Gamma'$ consists of a sequence of real and dummy vertices; the connected region $f$ of the plane that corresponds to $f'$ in $\Gamma$ consists of a sequence of vertices and crossing points. For simplicity we call $f$ a \emph{face} of $\Gamma$. The \emph{outer face} of $\Gamma$ is the face corresponding to the outer face of $\Gamma'$. A fan-planar drawing of $G$ with all vertices on the outer face is called an \emph{outer fan-planar drawing} of $G$. Observe that the configuration in Figure~\ref{fi:forbidden-conf-2} cannot occur in a drawing with all vertices on the outer face; hence, a drawing is outer fan-planar if and only if all vertices are on the outer face and it does not contain an edge crossed by two independent edges. An \emph{outer fan-planar graph} is a graph that admits an outer fan-planar drawing.

An outer fan-planar graph $G$ is \emph{maximal}, if no edge can be added to $G$ without loosing the property that $G$ remains outer fan-planar. An outer fan-planar graph $G$ with $n$ vertices is \emph{maximally dense} if it has the maximum number of edges among all outer fan-planar graphs with $n$ vertices. Notice that if $G$ is maximally dense then it is also maximal, but not vice-versa. The following property holds:

\begin{lemma}\label{le:outer-edges-maximal-outer-fan-planar}
Let $G=(V,E)$ be a maximal outer fan-planar graph and let $\Gamma$ be an outer fan-planar drawing of $G$.
The outer face of $\Gamma$ does not contain crossing points, i.e., it consists of $|V|$ uncrossed edges.   
\end{lemma}
\begin{proof}
Suppose by contradiction that the outer face of $\Gamma$ contains a crossing point $c$. Let $u$ be the first vertex of $G$ encountered while moving from $c$ counterclockwise along the boundary of the outer face of $\Gamma$ and let $v$ be the first vertex of $G$ encountered while moving from $c$ clockwise along the boundary of the outer face of $\Gamma$. One can add to $\Gamma$ a simple curve connecting $u$ and $v$ without crossing any other edge of $\Gamma$ and so that $u$ and $v$ remain on the outer face, thus contradicting the hypothesis that $G$ is maximal outer fan-planar.\qed     
\end{proof}

Given an outer fan-planar drawing $\Gamma$ of a maximal outer fan-planar graph $G$, the edges of $G$ that form the boundary of the outer face of $\Gamma$ will be also called the \emph{outer edges} of $\Gamma$.
 
\smallskip
A \emph{$2$-layer fan-planar drawing} is a fan-planar drawing such that: $(i)$ each vertex is drawn on one of two distinct horizontal lines, called \emph{layers}; $(ii)$ each edge connects vertices of different layers and it is drawn as a vertical monotone curve. By definition, a $2$-layer fan-planar drawing is also an outer fan-planar drawing. A \emph{$2$-layer fan-planar graph} is a graph that admits a $2$-layer fan planar drawing.

\section{Density of Outer and $2$-layer Fan-planar Graphs}\label{se:outerfan}

We first prove that an $n$-vertex outer fan-planar graph $G$ has at most $3n-5$ edges. Then we describe a family of outer fan-planar graphs with $n$ vertices and exactly $3n-5$ edges. These two results together give us a tight bound on the density of outer fan-planar graphs.

Let $G$ be a graph and let $\Gamma$ be a drawing of $G$. The \emph{crossing graph} of $\Gamma$, denoted as $\CG(\Gamma)$, is a graph having a vertex for each edge of $G$ and an edge between any two vertices whose corresponding edges cross in $\Gamma$. A cycle of $\CG(\Gamma)$ of odd length will be called an \emph{odd cycle} of $\CG(\Gamma)$; similarly, an \emph{even cycle} of $\CG(\Gamma)$ is a cycle of even length.  
In order to prove the $3n-5$ upper bound, we can assume that $G$ is a maximally dense outer fan-planar graph. We start by proving some interesting combinatorial properties of $G$ related to the cycles of the crossing graph of $G$. 

\begin{lemma}\label{le:3n-6}
Let $G=(V,E)$ be a maximal outer fan-planar graph with $n = |V|$ vertices and $m = |E|$ edges. Let $\Gamma$ be an outer fan-planar drawing of $G$. If $\CG(\Gamma)$ does not have odd cycles then $m \leq 3n-6$.
\end{lemma}
\begin{proof}
If $\CG(\Gamma)$ does not contain odd cycles, then it is bipartite and its vertices can be partitioned into two independent sets $W_1$ and $W_2$. Since by Lemma~\ref{le:outer-edges-maximal-outer-fan-planar} the outer edges of $\Gamma$ are not crossed, they correspond to $n$ isolated vertices in $\CG(\Gamma)$. We can arbitrarily assign all these vertices to the same set, say $W_1$. Denote by $E_i$ the set of edges of $G$ corresponding to the vertices of $W_i$ ($i \in \{1,2\}$). Clearly, $E_1$ and $E_2$  
partition the set $E$. Since no two edges of $E_i$ cross in $\Gamma$, then the two subgraphs $G_1 = (V,E_1)$ and $G_2=(V,E_2)$ are outerplanar graphs, where $|E_1| \leq 2n-3$ and $|E_2| \leq 2n-3-n$.
Thus, $m = |E| = |E_1| + |E_2| \leq 3n - 6$.\qed
\end{proof}

The next lemma shows that the length of any odd cycle of $\CG(\Gamma)$ is at most $5$. 

\begin{lemma}\label{le:oddcycles}
Let $G$ be a maximally dense outer fan-planar graph with $n$ vertices and let $\Gamma$ be an outer fan-planar drawing of $G$. $\CG(\Gamma)$ does not contain odd cycles of length greater than $5$.
\end{lemma}

\begin{proof}
Let $C$ be an odd cycle of length $\ell$ in $\CG(\Gamma)$. Let $E(C) = \{e_0=(u_0,v_0) \dots,e_{\ell-1}=(u_{\ell-1},v_{\ell-1})\}$ be the set of $\ell$ edges of $G$ corresponding to the vertices of $C$, such that 
$e_i$ crosses $e_{i+1}$ for $i=0,\dots,l-1$, where indices are taken modulo $\ell$.

Recall that all vertices of $G$ are on the outer face of $\Gamma$, which implies that the end-vertices of the edges in $E(C)$ are encountered in the following order when walking clockwise on the boundary of the outer face of $\Gamma$: $u_{i}$ precedes $v_{i-1}$ and $v_{i}$ precedes $u_{i+2}$ (see, e.g., Figure~\ref{fi:cycle-1}). Furthermore, vertices $v_{i}$ and $u_{i+2}$ must coincide, for $i=0,\dots,\ell-1$. Indeed, if $v_{i}$ and $u_{i+2}$ are distinct, for some $i=0,\dots,\ell-1$, then edge $e_{i+1}$ is crossed by two independent edges (i.e., $e_{i}$ and $e_{i+2}$), which contradicts the hypothesis that $\Gamma$ is fan-planar. See also Figure~\ref{fi:cycle-1}. Thus, we have that $u_{i}$ precedes $u_{i+1}$ while walking clockwise on the boundary of the outer face of $\Gamma$, for $i=0,\dots,\ell-1$, as shown in Figure~\ref{fi:cycle-2}. Moreover, it can be seen that the edges in $E(C)$ are not crossed by any edge not in $E(C)$, as otherwise the drawing would not be fan-planar. 

\begin{figure}
\centering
\subfigure[]{\includegraphics[width=0.35\columnwidth]{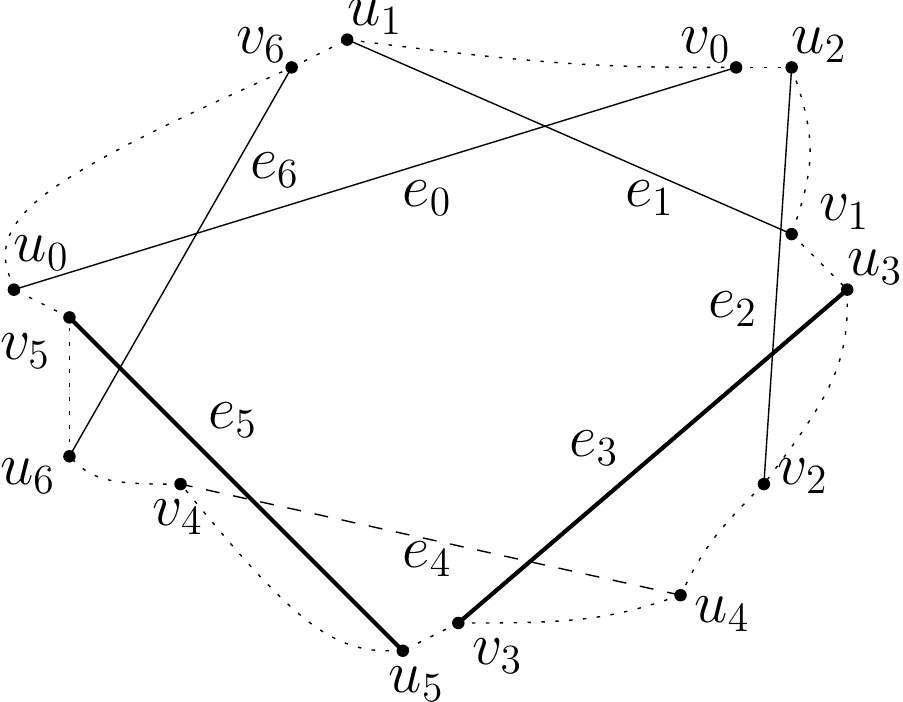}\label{fi:cycle-1}}
\hspace{2 cm}
\subfigure[]{\includegraphics[width=0.35\columnwidth]{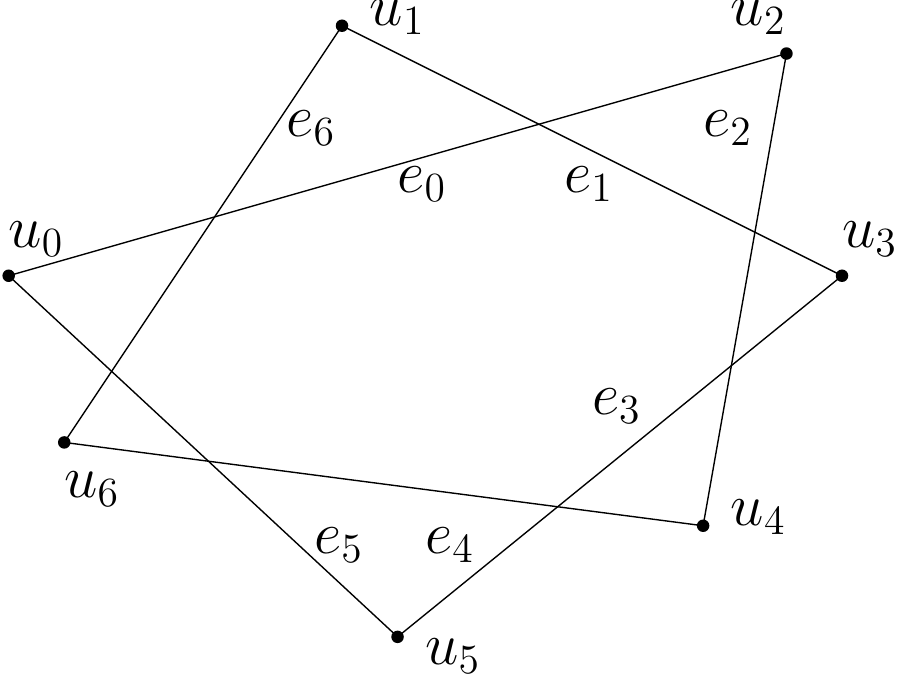}\label{fi:cycle-2}}
\caption{Illustration for the proof of Lemma~\ref{le:oddcycles}. (a) A set of edges $E(C)$ where $\ell=7$. It is possible to see that the dashed edge $e_4$ is crossed by the two bold edges $e_3$ and $e_5$, which are independent edges if $v_3$ and $u_5$ do not coincide. (b) A set of edges $E(C)$ where $\ell=7$, and where $v_{i}$ and $u_{i+2}$ always coincide, for $i=0,\dots,7$.}\label{fi:cycle}
\end{figure}

Now, suppose by contradiction that $\ell$ is odd and greater than $5$ (refer to Figure~\ref{fi:cycle-2} for an illustration). Consider a vertex $u_i$, for some $i=0,\dots,\ell-1$, and denote by $\overline{V}$ the set of vertices encountered between $u_{i+3}$ and $u_{i-3}$ while walking clockwise on the boundary of the outer face of $\Gamma$ (including $u_{i+3}$ and $u_{i-3}$). Vertex $u_i$ cannot be adjacent to any vertex in $\overline{V}$. Namely, if an edge $e=(u_i,u_j)$ existed, for some $u_j \in \overline{V}$, then it would be crossed by the two independent edges $e_{i-1}$ and $e_{j-1}$. Thus, removing $e_{i-1}$ from $\Gamma$, one can suitably connect $u_i$ to all the vertices in $\overline{V}$, still obtaining a fan-planar drawing $\Gamma^*$ with $n$ vertices. Since the size of $\overline{V}$ is $\ell-5$, and since $\ell \geq 7$ by assumption, we have that $\Gamma^*$ has at least two edges more than $\Gamma$, which contradicts the hypothesis that $G$ is maximally dense.\qed
\end{proof}

The following corollary is a consequence of Lemma~\ref{le:oddcycles} and Property~\ref{pr:3-mutually}.

\begin{corollary}\label{co:5-cycles}
Let $G$ be a maximally dense outer fan-planar graph. Any odd cycle in the crossing graph of a fan-planar drawing of $G$ has exactly length $5$. 
\end{corollary}

The following lemma proves that odd cycles in the crossing graph always correspond to $K_5$ subgraphs in $G$.

\begin{lemma}\label{le:k5}
Let $G$ be a maximally dense outer fan-planar graph, and let $\Gamma$ be an outer fan-planar drawing of $G$. If $\CG(\Gamma)$ contains a cycle $C$ of length $5$, then the subgraph of $G$ induced by the end-vertices of the edges corresponding to the vertices of $C$ is a $K_5$ graph.
\end{lemma}
\begin{proof}
Let $E(C) = \{e_0=(u_0,v_0),\dots,e_{4}=(u_{4},v_{4})\}$ be the set of $5$ edges of $G$ corresponding to the vertices of $C$, such that $e_i$ crosses $e_{i+1}$ for $i=0,\dots,4$, where indices are taken modulo $5$.

With the same argument used in the proof of Lemma~\ref{le:oddcycles}, vertices $v_{i}$ and $u_{i+2}$ must coincide, for $i=0,\dots,4$. It follows that $u_{i}$ precedes $u_{i+1}$ walking clockwise on the boundary of the outer face of $\Gamma$, and that $u_i$ is connected to $u_{i+2}$, for $i=0,\dots,4$. Moreover, $u_i$ and $u_{i+1}$ are connected by an edge, for $i=0,\dots,4$. Indeed, if there is no vertex of $G$ between $u_i$ and $u_{i+1}$ walking clockwise on the boundary of the outer face of $\Gamma$, for some $i=0, \dots, 4$, then the edge $(u_i,u_{i+1})$ can be added to $\Gamma$ without creating any crossing and so that all vertices remain on the outer face. If there is a vertex of $G$ between $u_i$ and $u_{i+1}$ walking clockwise on the boundary of the outer face of $\Gamma$ then it is easy to see that this vertex cannot be linked to any vertex $u_j$ distinct from $u_i$ and $u_{i+1}$, because this would cause a forbidden crossing (two independent edges crossed by an edge); it follows that edge $(u_i,u_{i+1})$ can be still added without creating crossing and so that all vertices of $G$ remain on the outer face. Hence, the subgraph induced by $u_0, u_1, \dots, u_4$ is $K_5$.\qed
\end{proof}

The next lemma proves the upper bound on the density of outer fan-planar graphs. 

\begin{lemma}\label{le:3n-5}
Let $G$ be a maximally dense outer fan-planar graph with $n$ vertices and $m$ edges. 
Then $m \leq 3n-5$ edges.
\end{lemma}
\begin{proof}
Let $\Gamma$ be an outer fan-planar drawing of $G$. We first claim that $G$ is biconnected. Suppose by contradiction that $G$ is not biconnected, and let $C_1$ and $C_2$ be two distinct biconnected components of $G$ that share a cut-vertex $v$. Let $u$ be the first vertex of $G$ encountered while moving from $v$ clockwise on the boundary of the outer face of $\Gamma[C_1]$, and let $w$ be the first vertex encountered while moving from $v$ counterclockwise on the boundary of the outer face of $\Gamma[C_2]$. It is possible to suitably add an edge $(u,w)$ in $\Gamma$, still getting an outer fan-planar drawing, which contradicts the hypothesis that $G$ is maximally dense. 

\smallskip 
Now, by Corollary~\ref{co:5-cycles}, $\CG(\Gamma)$ can only have either even cycles or cycles of length $5$. Also, by Lemma~\ref{le:k5}, every cycle of length $5$ in $\CG(\Gamma)$ corresponds to a subset of edges whose end-vertices induce $K_5$. We prove the statement by induction on the number $h$ of subgraphs of $G$ isomorphic to $K_5$. 

\smallskip\noindent{\em Base Case.} If $h=0$ then, by Lemma~\ref{le:3n-6}, $G$ has at most $3n-6$ edges, hence the statement holds. 

\smallskip\noindent{\em Inductive Case.} Suppose by induction that the claim is true for $h\geq 0$, and suppose $G$ contains $h+1$ subgraphs isomorphic to $K_5$. Let $G^*$ be one of these $h+1$ subgraphs. Let $e=(u,v)$ be an edge on the outer face of $\Gamma[G^*]$ that is not on the outer face of $\Gamma$. Vertices $u$ and $v$ are a separation pair of $G$, otherwise there would be a vertex of $G$ that is not on the outer face of $\Gamma$, which is impossible because $\Gamma$ is an outer fan-planar drawing by hypothesis. Hence, we can split $G$ into two biconnected subgraphs that share only edge $e$, one of them containing $G^*$.  

Let $G_1, G_2, \dots, G_k$ ($k \leq 5$) be the biconnected subgraphs of $G$ distinct from $G^*$ such that 
each $G_i$ shares exactly one edge with $G^*$. Each $G_i$ ($i=1,2,\dots,k$) contains at most $h$ subgraphs isomorphic to $K_5$, and therefore it has at most $3n_i-5$ edges by induction, where $n_i$ denotes the number of vertices of $G_i$. On the other hand, $G^*$ has $3n^*-5=10$ edges, where $n^*=5$ is the number of vertices of $G^*$. It follows that $m \leq 3(n^*+n_1+\dots+n_k)-5(k+1)-k$ ($k \leq 5$). Since $n^*+n_1+\dots+n_k \leq n+2k$ we have $m \leq 3(n+2k)-5(k+1)-k=3n-5$.\qed
\end{proof}

The existence of an infinite family of outer fan-planar graphs that match the $3n-5$ bound is proved in the next lemma. Refer to Figure~\ref{fi:lowerbound} for an illustration.

\begin{figure}[tb]
\centering
\subfigure[]{\includegraphics[width=0.25\columnwidth]{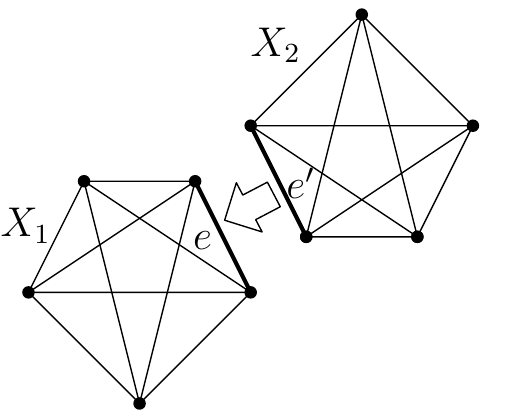}\label{fi:lowerbound-merge}}\hfill
\subfigure[]{\includegraphics[width=0.25\columnwidth]{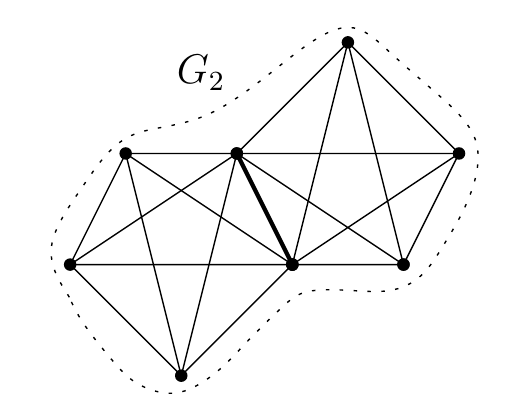}\label{fi:lowerbound-merge-2}}\hfill
\subfigure[]{\includegraphics[width=0.45\columnwidth]{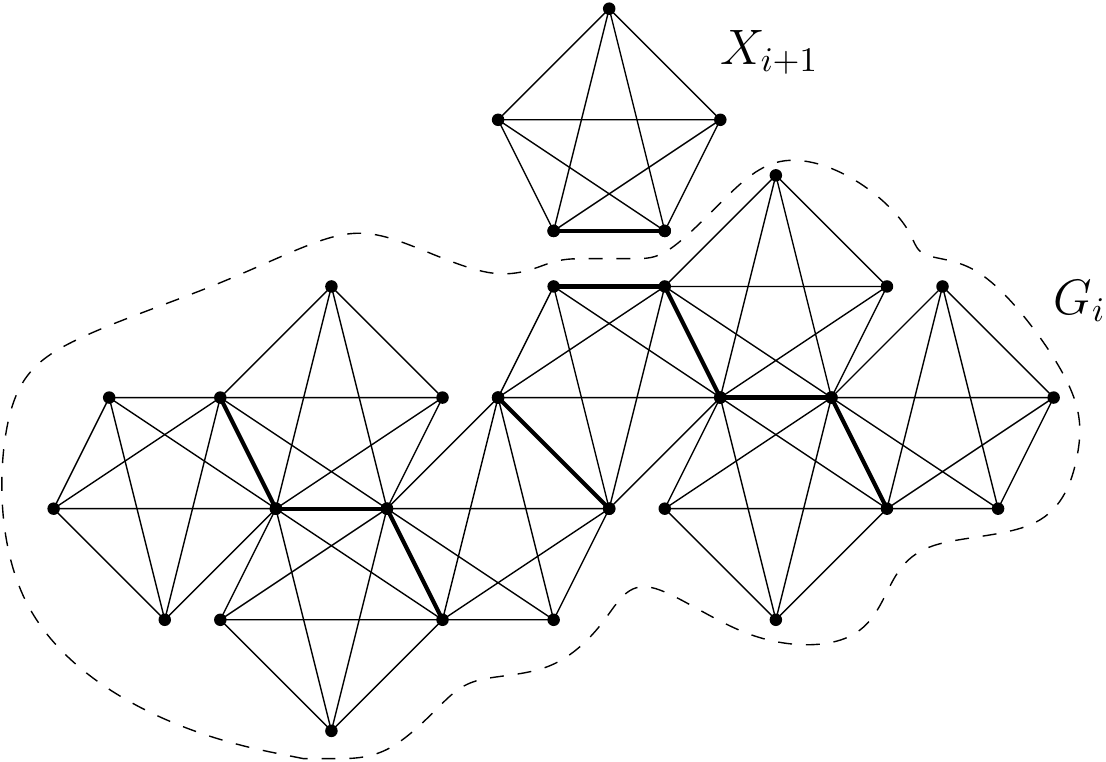}\label{fi:lowerbound-example}}
\caption{Illustration for the proof of Lemma~\ref{le:lowerbound}. (a) $X_1$ and $X_2$ before being merged. (b) Merging $X_1$ and $X_2$ into $G_2$. (c) $G_i$ and $X_{i+1}$, the bold edges are used for merging.}\label{fi:lowerbound}
\end{figure}

\begin{lemma}\label{le:lowerbound}
For any integer $h \geq 1$ there exists an outer fan-planar graph $G$ with $n=3h+2$ vertices and $m=3n-5$ edges.
\end{lemma}

\begin{proof}
Consider $h$ graphs $X_1,\dots,X_h$, such that $X_i$ is isomorphic to $K_5$, for $i=1,\dots,h$. 
We now describe how to construct $G$. The idea is to ``glue'' $X_1,\dots,X_h$ together in such a way that they share single edges one to another. The proof is by induction on the number of merged graphs. Denote by $G_i$ the graph obtained after merging $X_1,\dots,X_i$, for $1 < i \leq h$. We prove by induction that $G_i$ respects the following invariants: (I1) it is an outer fan-planar graph; (I2) it has $n_i=3i+2$ vertices and $m_i = 3n_i-5$ edges. In the base case $i=2$, we merge $G_1=X_1$ to $X_2$ as follows. Pick an edge $e$ on the outer face of $X_1$ and an edge $e'$ on the outer face of $X_2$. Merge $X_1$ and $X_2$ by identifying the edge $e$ with $e'$, see also Figures~\ref{fi:lowerbound-merge} and~\ref{fi:lowerbound-merge-2}. The new graph $G_2$ is clearly an outer fan-planar graph with $n_2=5+5-2=8$ vertices and $m_2=10+10-1=19$ edges. Thus, the two invariants hold. 

In the inductive case, suppose we constructed $G_i$, for $2<i<h$, and we want to attach $X_{i+1}$ (see also Figure~\ref{fi:lowerbound-example}). Pick any edge $e$ on the outer face of $G_i$ and any edge $e'$ on the outer face of $X_{i+1}$. Merge the two graphs in the same way as done in the base case. It is immediate to see that (I1) holds. Also, $n_{i+1} = n_i + 3$ and $m_{i+1}=m_i+9$. Since by induction $m_i = 3n_i - 5$, then $m_{i+1} = 3n_i - 5 + 9 = 3n_{i+1} - 5$.\qed
\end{proof}

Lemma~\ref{le:3n-5} and~\ref{le:lowerbound} imply the following theorem.

\begin{theorem}\label{th:densityouterfan}
An outer fan-planar graph with $n$ vertices has at most $3n-5$ edges, and this bound is tight for $n \geq 5$.
\end{theorem}

An obvious consequence of Theorem~\ref{th:densityouterfan} and of the definition of outer fan-planar graphs that are maximally dense is the following fact.

\begin{corollary}\label{co:maximallydense}
Every maximally dense outer fan-planar graph with $n=3h+2$ vertices $(h \geq 1)$ has $3n-5$ edges.
\end{corollary}

\remove{
We conclude this section by proving that if $G$ is a maximal outer fan-planar graph with $n$ vertices, then it has at least $2n$ edges. 

\begin{theorem}\label{th:maximal-mindensity}
Let $G$ be a maximal outer fan-planar graph with $n$ vertices and $m$ edges. Then $m \geq 2n$.
\end{theorem}

\begin{figure}[tb]
\centering
\subfigure{\includegraphics[width=0.3\columnwidth]{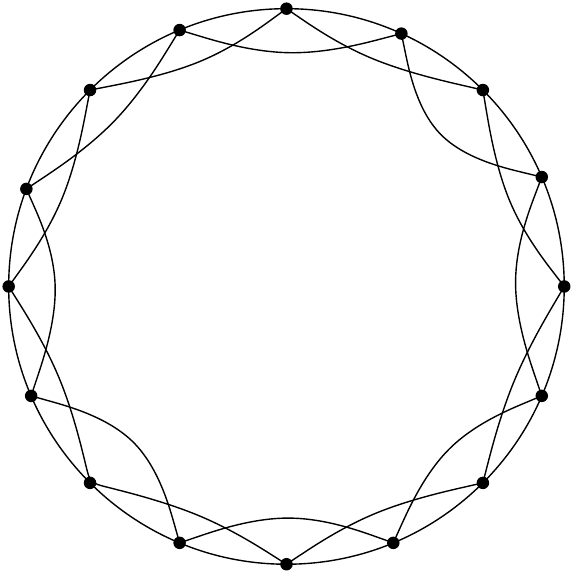}\label{fi:maximal-1}}\hfill
\subfigure{\includegraphics[width=0.3\columnwidth]{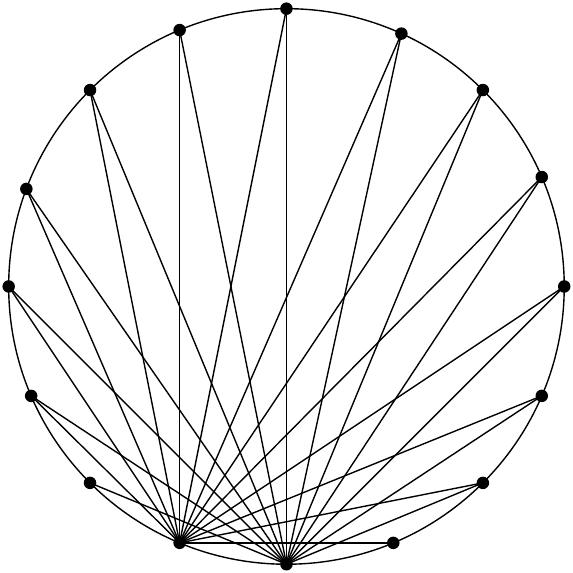}\label{fi:maximal-2}}
\caption{Illustration of the proof of Theorem~\ref{th:maximal-mindensity}}\label{fi:maximal}
\end{figure}
\todo[inline]{Here the idea is that we should focus on triconnected outer fan-planar graphs. Furthermore, we should prove that any triconnected maximal outer fan-planar graph $G$ with $n$ vertices is isomorphic to one of the following two graphs: $(i)$ the one in Figure~\ref{fi:maximal-1}; $(ii)$ the one in Figure~\ref{fi:maximal-2}. Observe that the graph in Figure~\ref{fi:maximal-1} has $2n$ edges, whereas the graph in Figure~\ref{fi:maximal-2} has $3n-5$ edges. Any idea on how to prove it? Maybe Luca can give us some hints, as this is part of the characterization that Luca et al. found during Bertinoro. Or maybe we should just say that we know this result by a private communication/arxiv paper by Luca et al.}
}

Concerning $2$-layer fan planar graphs, we already observed that a $2$-layer fan planar graph $G$ is an outer fan-planar graph. Also, since all vertices on the same layer form an independent set, graph $G$ is bipartite. We prove the following.

\begin{theorem}\label{th:2layer}
A $2$-layer fan-planar graph with $n$ vertices has at most $2n-4$ edges, and this bound is tight for $n \geq 3$.
\end{theorem}
\begin{proof}
Let $G$ be a maximally dense $2$-layer fan-planar graph with $n$ vertices and $m$ edges, and let $\Gamma$ be a $2$-layer fan-planar drawing of $G$. Denote by $V_1=\{v_1,\dots,v_{n_1}\}$ and $V_2=\{v_{n_1+1},\dots,v_{n}\}$ the two independent sets of vertices of $G$. Without loss of generality, suppose that in $\Gamma$ we have $v_i$ precedes $v_{i+1}$ along the layer of $V_1$ (for $i=1,\dots,n_1-1$), and $v_j$ precedes $v_{j+1}$ along the layer of $V_2$ (for $j=n_1+1,\dots,n-1$). See Figure~\ref{fi:2layer-proof}. Construct from $G$ a super-graph $G^*$, by adding an edge $(v_i,v_{i+1})$, for $i=1,\dots,n_1-1$, and an edge $(v_j,v_{j+1})$, for $j=n_1+1,\dots,n$ (see Figure~\ref{fi:2layer-proof-2}). Graph $G^*$ is still outer fan-planar. Moreover, $G$ does not contain a $K_5$ subgraph (because $G$ is bipartite), and this implies that also $G^*$ does not contain a $K_5$ subgraph, as otherwise at least three vertices of the same layer in $G$ should form a $3$-cycle in $G^*$ (which does not happen by construction). Thus, by Lemma~\ref{le:oddcycles} and Property~\ref{pr:3-mutually}, the crossing graph of any outer fan-planar drawing of $G^*$ contains only even cycles. Hence, denoted as $m^*$ the number of edges of $G^*$, by Lemma~\ref{le:3n-6} we have $m^* \leq 3n-6$, and therefore $m = m^* -(n-2) \leq 2n-4$.

\begin{figure}
\centering
\subfigure[]{\includegraphics[width=0.32\columnwidth]{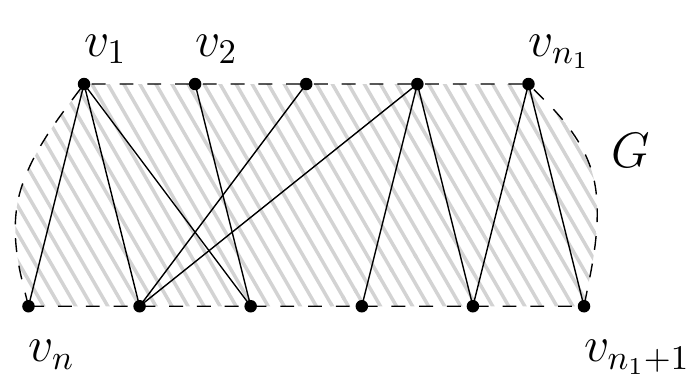}\label{fi:2layer-proof}}\hfill
\subfigure[]{\includegraphics[width=0.32\columnwidth]{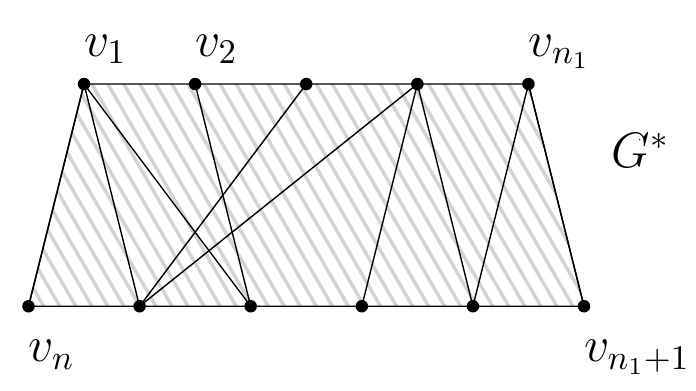}\label{fi:2layer-proof-2}}\hfill
\subfigure[]{\includegraphics[width=0.32\columnwidth]{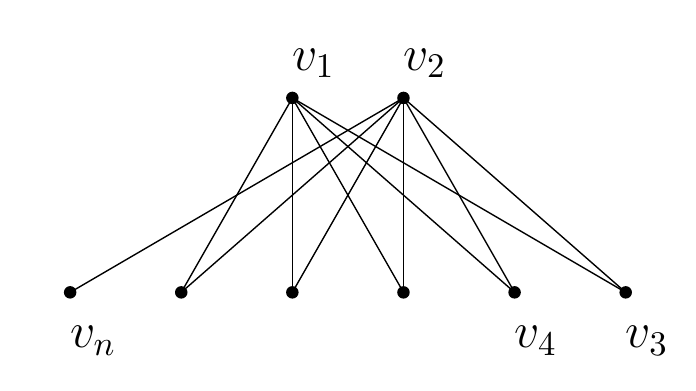}\label{fi:2layer-maximallydense}}
\caption{Illustration of the proof of Theorem~\ref{th:2layer}.}\label{fi:2layer}
\end{figure}

A family of $2$-layer fan-planar graphs with $2n-4$ edges is the family of the bipartite complete graphs $K_{2,n-2}$ (see Figure~\ref{fi:2layer-maximallydense}).\qed
\end{proof}

\section{Fan-planar and $k$-planar Graphs}\label{se:2planar}

A \emph{$k$-planar drawing} is a drawing where each edge is crossed at most $k$ times, and a \emph{$k$-planar graph} is a graph that admits a $k$-planar drawing. Clearly, every $1$-planar graph is also a fan-planar graph. Also, both the maximum number of edges of fan-planar graphs~\cite{DBLP:journals/corr/KaufmannU14} and the maximum number of edges of $2$-planar graphs~\cite{DBLP:journals/combinatorica/PachT97} have been shown to be $5n-10$. Thus it is natural to ask what is the relationship between fan-planar graphs and $2$-planar graphs and, more in general, what is the relationship between fan-planar and $k$-planar graphs for $k \geq 1$. In this section we show that there exist fan-planar graphs that are not $k$-planar, for every $k\geq1$, and that there are $k$-planar graphs (for $k>1$) that are not fan-planar. 

The existence of fan-planar graphs that are not $k$-planar can be proved using a counting argument on the minimum number of crossings of graph drawings. The \emph{crossing number} $cr(G)$ of a graph $G$ is the smallest number of crossings required in any drawing of $G$.

\begin{theorem}\label{th:fan-notk}
For every integer $k\geq1$ there exists a graph that is fan-planar but not $k$-planar.
\end{theorem}
\begin{proof}
Consider the complete $3$-partite graph $K_{1,3,h}$. It is easy to see that this graph is fan-planar for every $h \geq 1$ (see Figure~\ref{fi:k13h}). It is known that $cr(K_{1,3,h})=2 \left \lfloor \frac{h}{2} \right\rfloor  \left \lfloor \frac{h-1}{2} \right\rfloor + \left \lceil \frac{h}{2} \right\rceil$~\cite{JGT:JGT3190100102,s-gcnvs-13}. If we choose $h=4k+2$, we have $cr(K_{1,3,4k+2})=2 \left \lfloor \frac{4k+2}{2} \right\rfloor  \left \lfloor \frac{4k+1}{2} \right\rfloor + \left \lceil \frac{4k+2}{2} \right\rceil=4k(2k+1)+2k+1=8k^2+6k+1$. Thus, in every drawing of $K_{1,3,4k+2}$ there are at least $8k^2+6k+1$ crossings. On the other hand, in a $k$-planar drawing there can be at most $\frac{km}{2}$ crossings, where $m$ is the number of edges in the drawing. Since $K_{1,3,4k+2}$ has $16k+11$ edges, in order to be $k$-planar it should admit a drawing with at most $\frac{km}{2} = \frac{k(16k+11)}{2}=8k^2+\frac{11}{2}k$ crossings. Since $6k+1 > \frac{11}{2}k$ for every $k \geq 1$, $K_{1,3,4k+2}$ is not $k$-planar.\qed
\end{proof}

\begin{figure}
\centering
\subfigure[]{\label{fi:k13h}\includegraphics[scale=0.9]{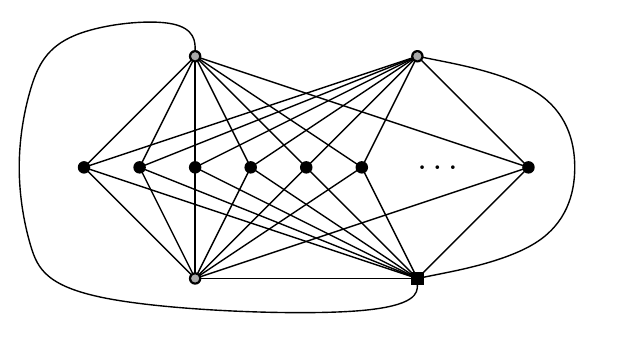}}
\subfigure[]{\label{fi:k7-fan-a}\includegraphics[scale=0.35]{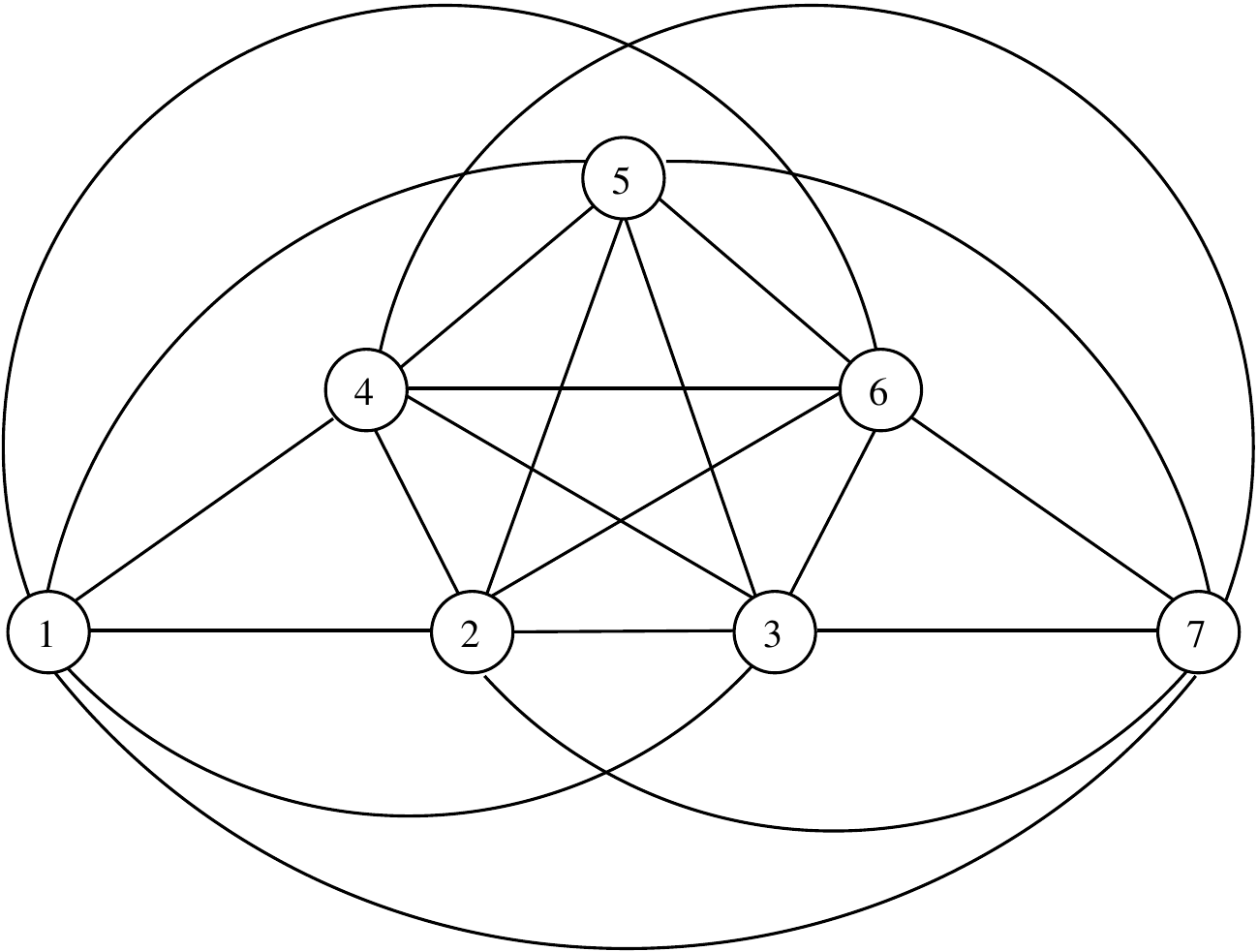}}
\hspace{0.2cm}
\subfigure[]{\label{fi:k7-fan-b}\includegraphics[scale=0.35]{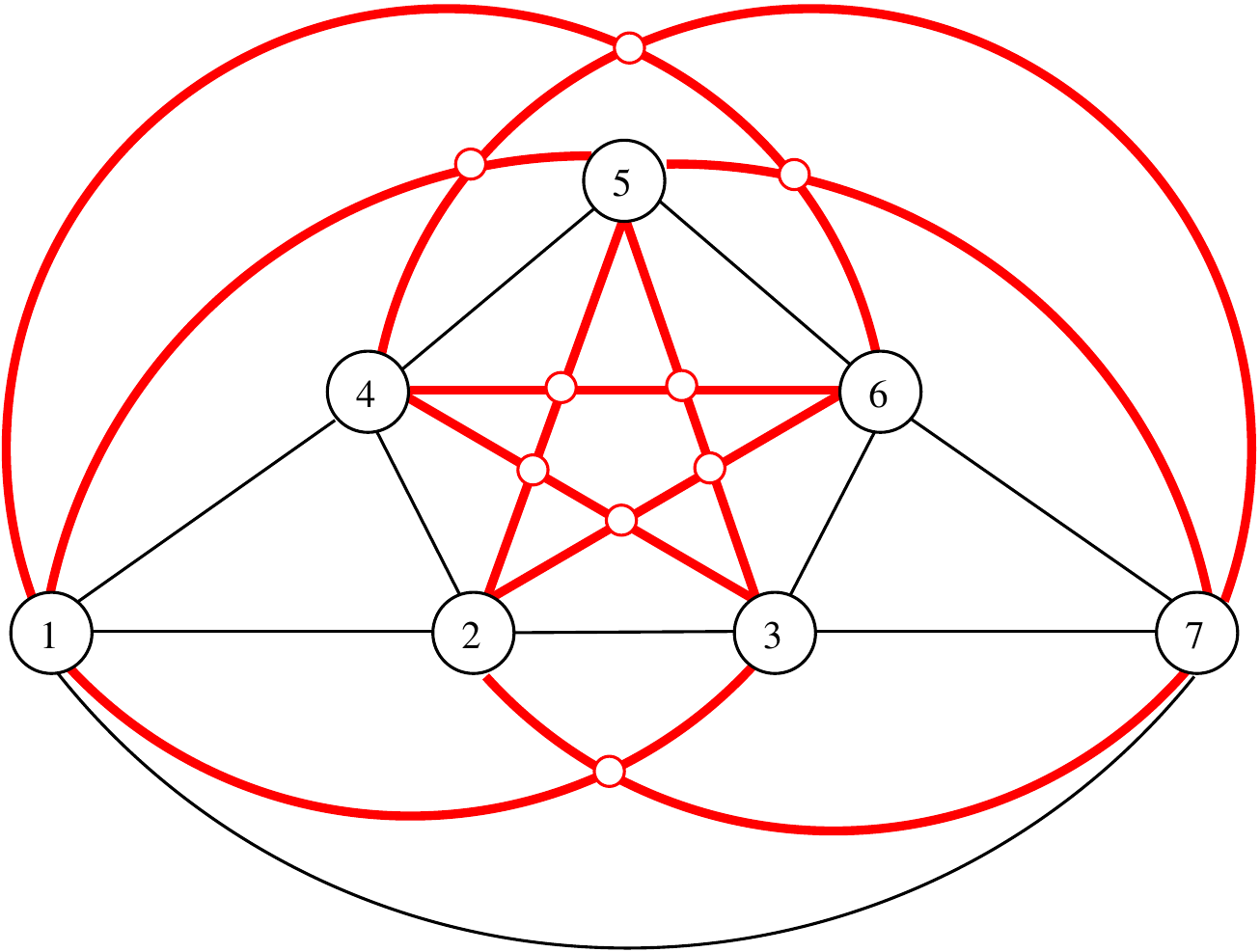}}
\caption{(a) A fan-planar drawing of $K_{1,3,h}$. (b) A fan-planar drawing of the $K_7$ graph. (c) The fragments of the fan-planar drawing in (a) are highlighted. }
\end{figure}


In order to prove that for any $k > 1$ there exist $k$-planar graphs that are not fan-planar (Theorem~\ref{th:k-notfan}), we need to prove a preliminary technical result (Lemma~\ref{le:k7}), which will be also reused in Section~\ref{se:hardness}. 

Let $\Gamma$ be a fan-planar drawing of a graph. We may regard crossed edges of $\Gamma$ as composed by \emph{fragments}, where a fragment is the portion of the edge that is between two consecutive crossings or between one of the two end-vertices of the edge and the first crossing encountered while moving along the edge towards the other end-vertex. An edge that is not crossed does not have any fragment. Figure~\ref{fi:k7-fan-a} shows a fan-planar drawing of the $K_7$ graph and Figure~\ref{fi:k7-fan-b} shows the fragments of the drawing in Figure~\ref{fi:k7-fan-a}. We consider two fragments \emph{adjacent} if they share a common crossing or a common end-vertex. The following lemma shows an interesting property of the fragments of any fan-planar drawing of the $K_7$ graph.


\begin{lemma}\label{le:k7}
In any fan-planar drawing of the $K_7$ graph, any pair of vertices is joined by a sequence of adjacent fragments. 
\end{lemma}
\begin{proof}
Consider a fan-planar drawing $\Gamma$ of the $K_7$ graph and consider any vertex $v_i$ of it. Vertex $v_i$ must be incident to some fragment in $\Gamma$. Indeed, if vertex $v_i$ had no incident fragment, all the edges incident to $v_i$ were uncrossed in $\Gamma$, and removing $v_i$ and all its incident edges from $\Gamma$ would yield a fan-planar drawing of the $K_6$ graph where all vertices are on the same face; this would clearly imply the existence of a fan-planar drawing of $K_6$ where all vertices are on the outer face, i.e., an outer fan-planar drawing of $K_6$. This is however impossible by Lemma~\ref{le:3n-5} ($K_6$ has $6$ vertices and $15$ edges, i.e., more than $3\cdot6-5$ edges). Since a fragment is originated by a crossed edge and since two crossing edges are not adjacent, we have that vertex $v_i$ is linked by a sequence of fragments to at least other three distinct vertices. Therefore, the vertices of $K_7$ are linked by sequences of fragments in groups of at least four. Being seven vertices in total, this implies that all vertices of $K_7$ are linked together by sequences of fragments.\qed
\end{proof}

\begin{figure}
\centering
\subfigure[]{\label{fi:kplanar-non-fanplanar-1}\includegraphics[scale=1,page=6]{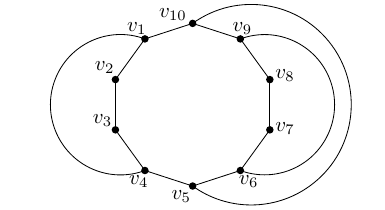}}
\subfigure[]{\label{fi:kplanar-non-fanplanar-2}\includegraphics[scale=1,page=5]{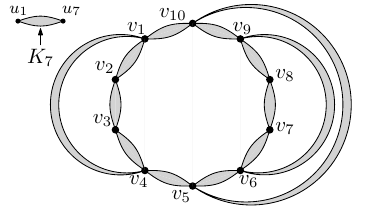}}
\subfigure[]{\label{fi:kplanar-non-fanplanar-3}\includegraphics[scale=1,page=4]{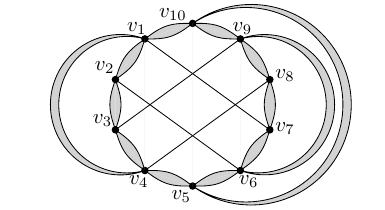}}
\caption{(a)--(c) Illustration for the proof of Theorem~\ref{th:k-notfan}: (a) the graph $G'$; (b) the graph $G''$; (c) the graph $G$. }
\end{figure}

\begin{theorem}\label{th:k-notfan}
For every integer $k>1$ there exists a graph that is $k$-planar but not fan-planar.
\end{theorem}
\begin{proof}
Since every $2$-planar graph is also a $k$-planar graph, for $k>1$, it is sufficient to prove  
that there exists a $2$-planar graph that is not fan-planar. 
Let $G'$ be a graph consisting of a cycle $C=(v_1,v_2,\dots, v_{10})$ and the three edges $(v_1,v_4)$, $(v_5,v_{10})$, and $(v_6,v_9)$ (see Figure~\ref{fi:kplanar-non-fanplanar-1}). Let $G''$ be the graph obtained from $G'$ by replacing each edge $(v_i,v_j)$ ($1 \leq i,j \leq 10$) with a copy of the $K_7$ graph (whose vertices are denoted as $u_1,u_2,\dots,u_7$) so that $v_i=u_1$ and $v_j=u_7$ (see Figure~\ref{fi:kplanar-non-fanplanar-2}). The copy of $K_7$ that replaces $(v_i,v_j)$ will be denoted as $K^{i,j}_7$. Let $G$ be the graph obtained from $G''$ by adding the four edges $(v_1,v_7)$, $(v_2,v_6)$, $(v_3,v_9)$, and $(v_4,v_8)$ (see Figure~\ref{fi:kplanar-non-fanplanar-3}). Graph $G$ is $2$-planar. Namely, planarly embed $G'$ as shown in Figure~\ref{fi:kplanar-non-fanplanar-1}. Construct a drawing $\Gamma$ of $G$ by replacing each edge of $G'$ with a drawing of $K^{i,j}_7$ like the one shown in Figure~\ref{fi:k7-fan-a} (see Figure~\ref{fi:kplanar-non-fanplanar-2}), and 
draw the four edges $(v_1,v_7)$, $(v_2,v_6)$, $(v_3,v_9)$, and $(v_4,v_8)$ inside the cycle $C$ as shown in Figure~\ref{fi:kplanar-non-fanplanar-3}. Drawing $\Gamma$ is clearly $2$-planar.

We now prove that $G$ is not fan-planar. Suppose by contradiction that $G$ has a fan-planar drawing $\Gamma$. By Lemma~\ref{le:k7}, for each $K^{i,j}_7$ ($1 \leq i,j \leq 10$) there is a sequence of fragments leading from $v_i=u_1$ to $v_j=u_7$; call such a sequence of fragments the \emph{spine} of $K_7^{i,j}$.  Delete from $\Gamma$ all fragments except those in the spine of each $K_7^{i,j}$; delete also all non-crossed edges and isolated vertices. The remaining drawing $\Gamma'$ is a planar drawing, because each spine cannot be crossed by any other fragment or edge, otherwise the drawing is no longer fan-planar. We denote by $C'$ the cycle of spines corresponding to $C$, by $S$ the set of spines of $K_7^{1,4}$, $K_7^{5,10}$, and $K_7^{6,9}$, and by $F$ the set of edges $(v_1,v_7)$, $(v_2,v_6)$, $(v_3,v_9)$, and $(v_4,v_8)$. Since $\Gamma'$ is planar each spine in $S$ is either inside $C'$ or outside $C'$ in $\Gamma'$, and therefore in $\Gamma$. Furthermore, since the edges of $F$ cannot cross the spine of any $K_7^{i,j}$ ($1 \leq i,j \leq 10$) in $\Gamma$, each of them must be either inside or outside $C'$ in $\Gamma$. Given two elements of $S \cup F$ we say that they are on the \emph{same side} of $C'$ if they are both inside or both outside $C'$ in $\Gamma$, otherwise we say that they are on \emph{opposite sides} of $C'$. 
Since there cannot be a crossing between an element of $F$ and one of $S$, each of the two edges $(v_2,v_6)$ and $(v_3,v_9)$ must be on the opposite side of $C'$ with respect to $K_7^{1,4}$. Analogously, each of the two edges $(v_1,v_7)$ and $(v_4,v_8)$ must be on the opposite side of $C'$ with respect to $K_7^{6,9}$. Finally, $K_{7}^{5,10}$ must be on the opposite side of $C'$ with respect to $(v_1,v_7)$, $(v_2,v_6)$, $(v_3,v_9)$, and $(v_4,v_8)$. It follows that the spines of $S$ and the edges of $F$ must be on opposite sides of $C'$, which implies that each edge in $F$ is crossed by two independent edges (see Figure~\ref{fi:kplanar-non-fanplanar-3}), a contradiction.\qed
\end{proof}
\section{Complexity of the Fan-planarity Testing Problem}\label{se:hardness}

In this section, we exploit the results of Section~\ref{se:outerfan} and Section~\ref{se:2planar} to prove that testing whether a graph is fan-planar in the variabile embedding setting is NP-complete. We simply refer to this problem as the \emph{fan-planarity testing}. Our proof is based on a reduction from \emph{1-planarity testing}, which is known to be NP-complete in the variable embedding setting~\cite{DBLP:journals/algorithmica/GrigorievB07,DBLP:journals/jgt/KorzhikM13}. The 1-planarity testing asks whether a given graph admits a 1-planar drawing. We prove the following. 

\begin{theorem}\label{th:hardness}
Fan-planarity testing is NP-complete. 
\end{theorem}
\begin{proof}
First we prove that the problem is NP-hard and then we prove that it belongs to NP.

Given an instance $G=(V,E)$ of the 1-planarity testing we build an instance $G_f=(V_f,E_f)$ of the fan-planarity testing by replacing each edge $(u,v) \in E$ with two $K_7$ graphs with vertices $u=u_{1}, u_{2}, \dots, u_{7}$ and $v=v_{1}, v_{2}, \dots, v_{7}$, called \emph{attachment gadgets} and joined by a \emph{spanning} edge $(u_7,v_7)$. See Figure~\ref{fi:hardness-reduction} for a schematic illustration of this reduction. $G_f=(V_f,E_f)$ can be constructed in polynomial time, having $|V_f| = |V| + |E| \times 12$ vertices and $|E_f| = |E| \times 43$ edges, where $|E| \times 42$ of them belong to the attachment gadgets and the remaining $|E|$ are spanning edges that join different attachment gadgets.

\begin{figure}
\centering
\subfigure[]{\includegraphics[width=0.32\columnwidth]{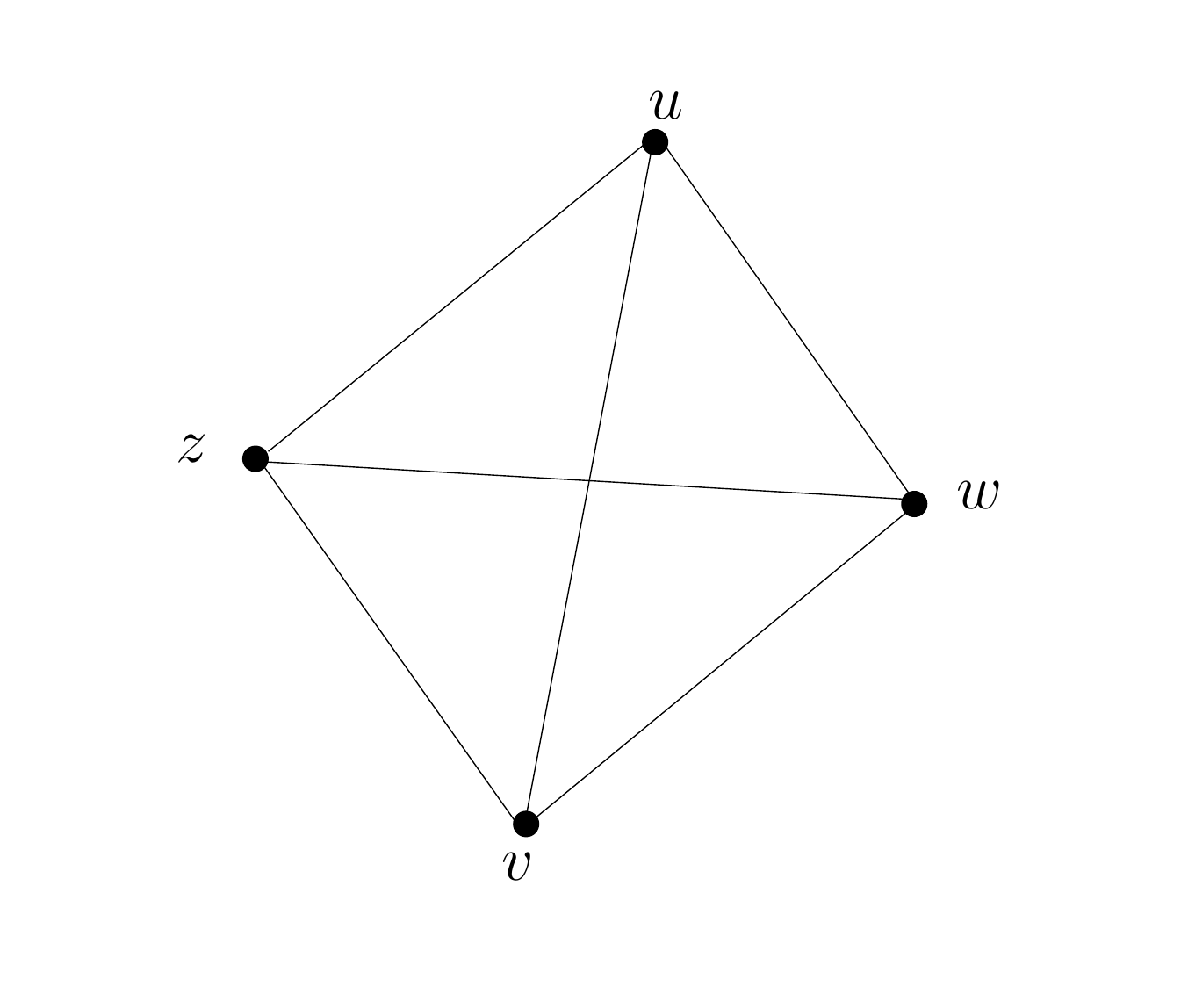}\label{fi:hardness-reduction-a}}
\subfigure[]{\includegraphics[width=0.32\columnwidth]{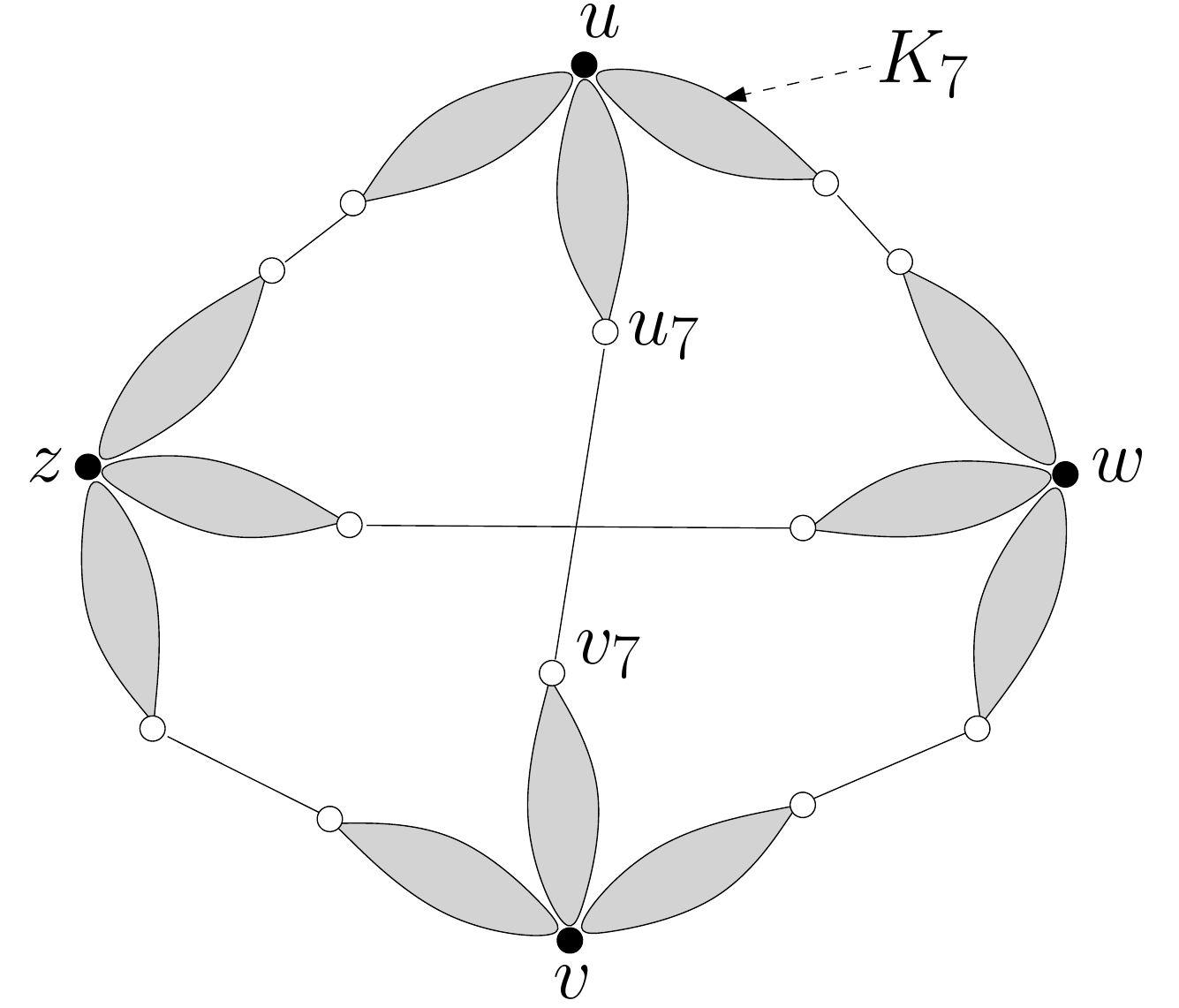}\label{fi:hardness-reduction-b}}
\subfigure[]{\includegraphics[width=0.32\columnwidth]{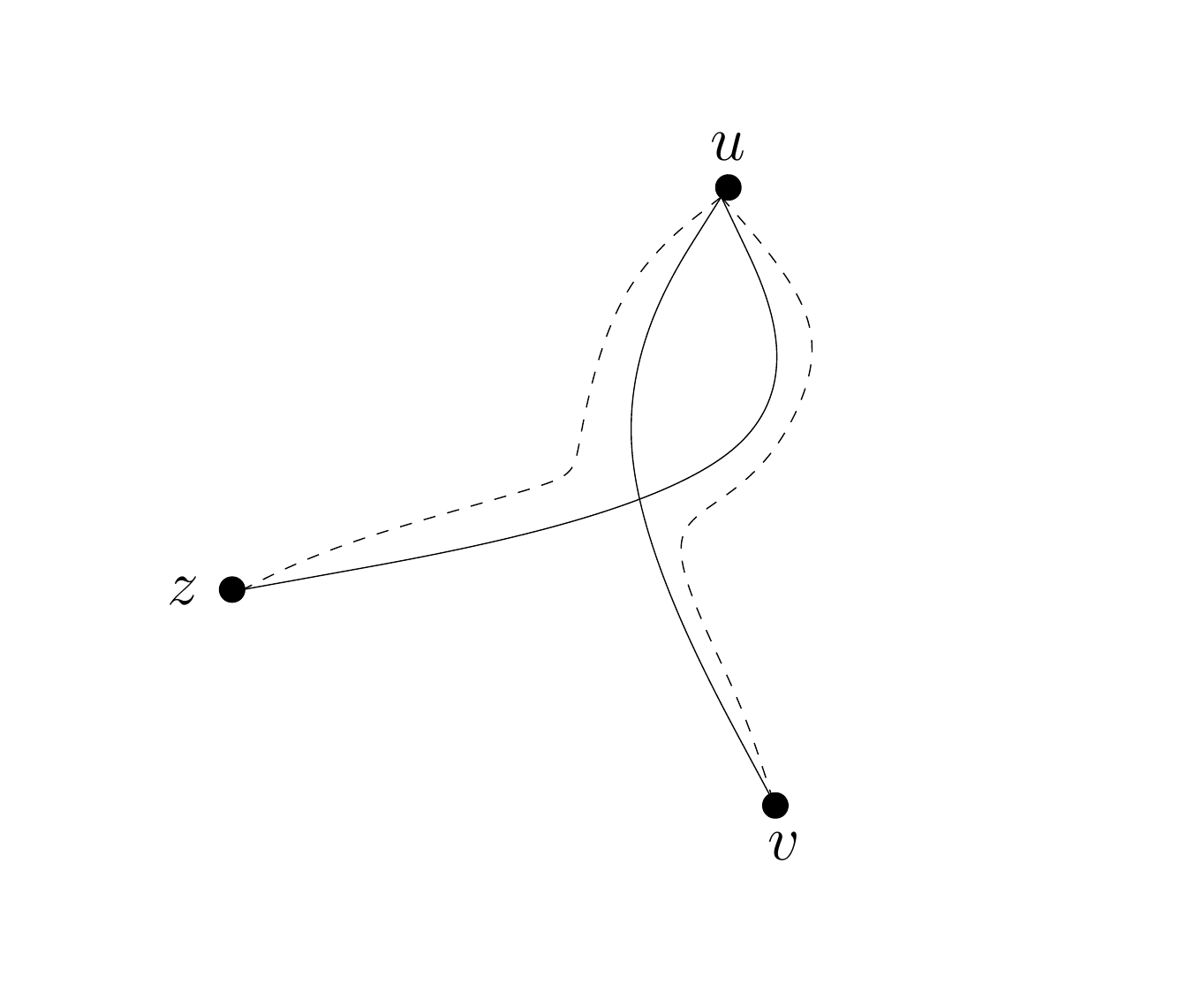}\label{fi:hardness-reduction-c}}
\caption{Illustration of the reduction in the proof of Theorem~\ref{th:hardness}. (a) An instance $G$ of 1-planarity testing; (b) The reduced instance $G_f$ of fan-planarity testing. (c) Two edges of $G$ that are adjacent to a common vertex $u$ and such that they cross one to another in $\Gamma$. Such a crossing can be removed by rerouting the two edges as shown by the dashed paths.  }\label{fi:hardness-reduction}
\end{figure}

We now prove that $G$ is 1-planar if and only if $G_f$ is fan-planar. If $G$ admits a 1-planar drawing, replace each edge $(u,v)$ of $G$ with two fan-planar drawings of $K_7$ like those depicted in Fig.~\ref{fi:k7-fan-a} and with edge $(u_7,v_7)$, in such a way that the possible crossing of $(u,v)$ occurs on $(u_7,v_7)$. The obtained drawing of $G_f$ is fan-planar since each attachment gadget has a fan-planar drawing and each spanning edge has at most one crossing. Conversely, suppose $G_f$ admits a fan-planar drawing $\Gamma_f$. By Lemma~\ref{le:k7}, for any attachment gadget of $G_f$ attached to vertex $u$, there is at least a sequence of fragments leading from $u=u_1$ to $u_7$. As in the proof of Theorem~\ref{th:k-notfan}, call such a sequence of fragments the \emph{spine} of the attachment gadget. Delete from $\Gamma_f$ all fragments except those in the spines. Delete from $\Gamma_f$ all uncrossed edges except the spanning edges. Remove also isolated vertices. A drawing $\Gamma$ of $G$ is obtained, where the drawing of edge $(u,v)$ is given by the spine from $u=u_1$ to $u_7$, the spanning edge $(u_7,v_7)$, and the spine from $v_7$ to $v_1=v$. Observe that, $u \neq v$, as otherwise there would be a self-loop in $G$. We claim that $\Gamma$ is a 1-planar drawing of $G$. Indeed, fragments in the spines can not be crossed by any other fragment or spanning edge of $\Gamma_f$. It follows that spanning edges can cross only among themselves in $\Gamma_f$. However, they can cross only once, as they are a matching of $G_f$ and $\Gamma_f$ is fan-planar. Hence, $\Gamma$ is a 1-planar drawing, but not necessarily simple; indeed, it may happen that two crossing edges $(u,v)$ and $(w,z)$ in $\Gamma$ share an end-vertex, say $u=w$ (this happens when in $\Gamma_f$ there are two crossing spanning edges of two $K_7$ attached to $u$). However, the crossing between $(u,v)$ and $(u,z)$ in $\Gamma$ can be easily removed by rerouting the two edges as shown in Fig.~\ref{fi:hardness-reduction-c}.

We now prove that the fan-planarity testing is in NP. A non-deterministic Turing machine can be devised to test if a graph $G$ with $n$ vertices and $m$ edges admits a fan-planar drawing. The main strategy consists in exploiting non-determinism to explore all drawings of $G$ with $k$ crossings, where $0 \leq k \leq \binom{m}{2}$. Analogously to~\cite{GJ-cnNPc-83} a simple algorithm to test if a graph admits a fan-planar drawing with $k$ crossings non-deterministically considers all possible $k$ pairs of edges that cross (and the order in which crossings occur along edges involved in more than one crossing), discards the configurations where there is an edge that crosses more than one fan, replaces crossings with dummy vertices, and tests the obtained graph for planarity. \qed

\end{proof}

\section{Conclusions and Open Problems}\label{se:conclusions}

We extended the study of fan-planar drawings started by Kaufmann and Ueckerdt~\cite{DBLP:journals/corr/KaufmannU14}. We showed tight bounds on the density of constrained versions of fan-planar drawings and clarifed the relationship between fan-planarity and $k$-planarity. Also, we proved that the fan-planarity testing in the variable embedding setting is NP-complete. Several interesting problems remain open. We mention few of them, in addition to those already listed in~\cite{DBLP:journals/corr/KaufmannU14}:

\begin{itemize}
\item What is the minimum number of edges of maximal fan-planar graphs?
\item What is the complexity of deciding whether a graph with a given rotation system is fan-planar?
\item Can we efficiently recognize maximally dense fan-planar graphs? 
\end{itemize}

\section*{Acknowledgments}
This work started at the Bertinoro Workshop on Graph Drawing 2014.
We thank Michael Kaufmann and Torsten Ueckerdt for suggesting the study of fan-planar graphs during the workshop. 
We also thank all the participants of the workshop for the useful discussions on this topic.    

{\small \bibliography{fanplanar}}
\bibliographystyle{splncs03}

\end{document}